\documentclass[preprint]{elsarticle}

\usepackage{lineno,hyperref}
\modulolinenumbers[5]

\makeatletter
\def\ps@pprintTitle{%
	\let\@oddhead\@empty
	\let\@evenhead\@empty
	\def\@oddfoot{\reset@font\hfil\thepage\hfil}
	\let\@evenfoot\@oddfoot
}
\makeatother

\usepackage{amsmath}
\usepackage{tabularx}
\usepackage{tablefootnote}
\usepackage{array}

\usepackage{multirow}
\usepackage{enumitem}
\setlist[itemize]{noitemsep, topsep=0pt}
\setlist[enumerate]{noitemsep, topsep=0pt}
\usepackage{color,soul}

\DeclareMathAlphabet{\mathpzc}{OT1}{pzc}{m}{it}
\usepackage{amsthm,amsmath,amssymb}

\newtheorem{proposition}{Proposition}

\theoremstyle{definition}

\theoremstyle{remark}
\newtheorem{remark}{Remark}

\begin{document}

\begin{frontmatter}
	
	\title{On the Security of an Unconditionally Secure, Universally Composable Inner Product Protocol}

	\author{Babak~Siabi}
	\ead{b.siabi@ec.iut.ac.ir}
	
	\author{Mehdi~Berenjkoub}
	\ead{brnjkb@cc.iut.ac.ir}

	\begin{abstract}
		In this paper we discuss the security of a distributed inner product (DIP) protocol [\textit{IEEE TIFS}, 11(1), (2016), 59-73]. We show information leakage in this protocol that does not happen in an ideal execution of DIP functionality. In some scenarios, this information leakage enables one of the parties to completely learn the other party’s input. We will give examples of such scenarios.
	\end{abstract}
	
	\begin{keyword}
		inner product\sep preprocessing model\sep secure computation\sep security analysis
	\end{keyword}
	
\end{frontmatter}



\section{Preliminaries}\label{sec:preliminaries}

\subsection{Notations}
In accordance with \cite{david2016unconditionally}, in the following, we denote by $\mathbb{F}_q$ the finite field of order $ q $, by ${\mathbb{F}_q}^n$ the space of all $n$-vector with elements in $\mathbb{F}_q$ and by ${\mathbb{F}_q}^{m\times n}$ the space of all $m\times n$ matrices with elements belonging to $\mathbb{F}_q$. We use over-bar lowercase letters to represent vectors in ${\mathbb{F}_q}^n$ (e.g. $\bar{a}$ represents a vector), and bold uppercase letters to represent matrices in ${\mathbb{F}_q}^{m \times n}$ (e.g. $\textbf{A}$ represents a matrix). We denote by $x \in_R D$ the process of uniformly random sampling of element $x$ from domain $D$.

\subsection{Distributed inner product functionality}
The functionality considered in \cite{david2016unconditionally}, is the distributed version of two-party inner product (we refer to this functionality as DIP). In contrast to the conventional inner product, in DIP the result is shared between parties. More precisely, it is assumed that $P_1$ and $P_2$ hold private vectors $\bar{x}_1\in {\mathbb{F}_q}^k$ and $\bar{x}_2\in {\mathbb{F}_q}^k$, respectively, and intend to calculate $w=\langle \bar{x}_1 \cdot \bar{x}_2 \rangle$ such that for $i=1,2$ party $P_i$ receives an additive random share $w_i\in \mathbb{F}_q$ satisfying $w_1+w_2=w$.

In an ideal world, this functionality is handled by a TTP as illustrated in Fig. \ref{fig:ideal functionality}. Since the TTP honestly follows the illustrated procedure, the output shares are uniformly distributed on $\mathbb{F}_q$ and parties learn nothing about $w=\langle \bar{x}_1 \cdot \bar{x}_2 \rangle$ or the input of the other party.

\begin{figure}[!t]
	\fbox{
		\begin{minipage}{0.95\columnwidth}
			\begin{center}
				\textbf{Functionality $\mathcal{F}_{DIP}$}
				\rule{\linewidth}{1pt}
			\end{center}

			\textbf{Inputs:} $P_1$ and $P_2$ hold $\bar{x}_1\in {\mathbb{F}_q}^k$ and $\bar{x}_2\in {\mathbb{F}_q}^k$ respectively.
			
			\textbf{Output:} $P_1$ learns $w_1\in_R \mathbb{F}_q$ and $P_2$ learns $w_2=\langle \bar{x}_1 \cdot \bar{x}_2 \rangle - w_1$.
			
			\rule{\linewidth}{1pt}
			
			$P_1:$ Sends $\bar{x}_1$ to TTP.
			
			$P_1$: Sends $\bar{x}_2$ to TTP.
			
			TTP: Upon receiving $\bar{x}_1$ and $\bar{x}_2$,
			\begin{itemize}
				\item if $\bar{x}_1 \notin {\mathbb{F}_q}^k$ or $\bar{x}_2 \notin {\mathbb{F}_q}^k$ sets $w_1=w_2=\lambda$,

				\item otherwise, chooses $u\in _R \mathbb{F}_q$ , sets $w_1=u$ and $w_2=\langle \bar{x}_1\cdot \bar{x}_2 \rangle -u$.
				
			\end{itemize}
			TTP: Sends $w_1$ to $P_1$ and $w_2$ to $P_2$.
		\end{minipage}
		}
	\caption{The ideal functionality of distributed inner product ($\mathcal{F}_{DIP}$)
						\label{fig:ideal functionality}}
\end{figure}

\begin{remark}\label{Rem:SSPandCSP}
	It is of vital importance to note the difference between DIP and conventional inner product. In the latter, since the parties learn the product, they always (even in the ideal world) can derive an equation for the other party’s input using their own input and output. But, derivation of such an equation is not possible in the case of DIP ideal functionality.
\end{remark}

\section{DIP Protocol of \cite{david2016unconditionally}}\label{sec:DIP protocol}
In \cite{david2016unconditionally}, David \emph{et al.} propose a protocol to realize $\mathcal{F}_{DIP}$. Fig. \ref{fig:DDG+ SSP protocol} illustrates this protocol (Protocol $\pi _{DDG+}$). Protocol $\pi _{DDG+}$ is designed in the preprocessing model and planned to be universally composable and unconditionally secure. In the preprocessing model, it is assumed that an initiator (denoted by $Init.$ in Fig. \ref{fig:DDG+ SSP protocol}) distributes a set of correlated randomness between parties before they decide (or fix) their inputs (preprocessing phase). After deciding the inputs, parties compute the functionality with the aid of preprocessed data (computation phase).

\begin{figure}[!t]
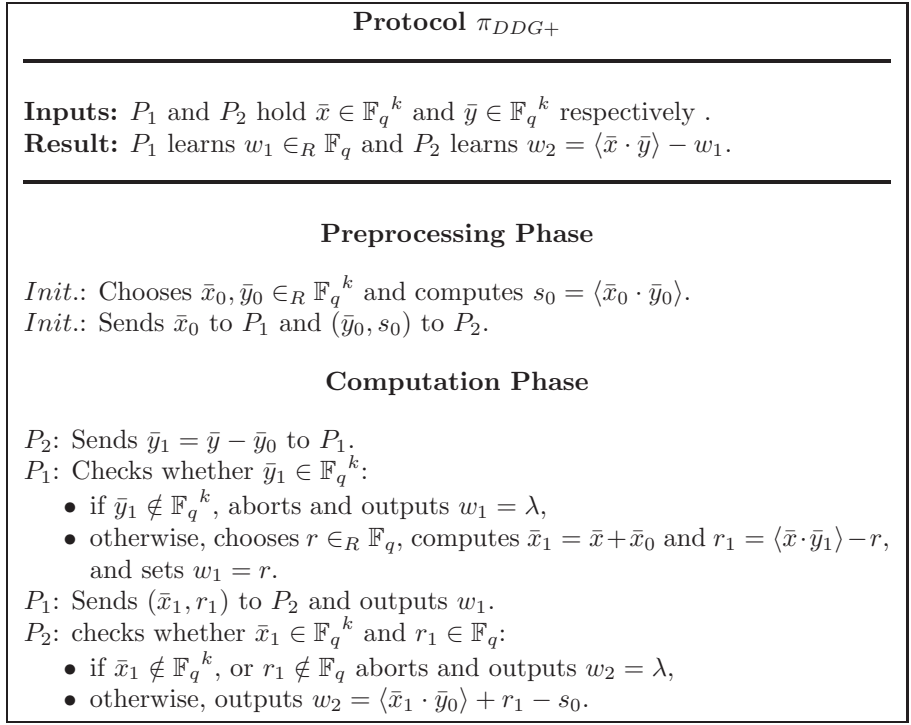

	\fbox{
		\begin{minipage}{0.95\columnwidth}
			\begin{center}
				\textbf{Protocol $\pi_{DDG+}$}
				\rule{\linewidth}{1pt}
			\end{center}

			\textbf{Inputs:} $P_1$ and $P_2$ hold $\bar{x}\in {\mathbb{F}_q}^k$ and $\bar{y}\in {\mathbb{F}_q}^k$ respectively .
			
			\textbf{Result:} $P_1$ learns $w_1\in_R \mathbb{F}_q$ and $P_2$ learns $w_2=\langle \bar{x} \cdot \bar{y} \rangle - w_1$.
			
			\rule{\linewidth}{1pt}

			\begin{center}
				\textbf{Preprocessing Phase}
			\end{center}
			
			$Init.$: Chooses $\bar{x}_0,\bar{y}_0 \in_R {\mathbb{F}_q}^k$ and computes $s_0= \langle \bar{x}_0 \cdot \bar{y}_0 \rangle$.
			
			$Init.$: Sends $\bar{x}_0$ to $P_1$ and ($\bar{y}_0, s_0$) to $P_2$.
						
			\begin{center}
				\textbf{Computation Phase}
			\end{center}
			
			$P_2$: Sends $\bar{y}_1= \bar{y}-\bar{y}_0$ to $P_1$.
		
			$P_1$: Checks whether $\bar{y}_1\in {\mathbb{F}_q}^k$:
			\begin{itemize}
				\item if $\bar{y}_1\notin {\mathbb{F}_q}^k$, aborts and outputs $w_1=\lambda$,
				\item otherwise, chooses $r\in_R \mathbb{F}_q$, computes $\bar{x}_1= \bar{x}+\bar{x}_0$ and $r_1= \langle \bar{x} \cdot \bar{y}_1 \rangle -r$, and sets $w_1=r$.
			\end{itemize}
			$P_1$: Sends $(\bar{x}_1, r_1)$ to $P_2$ and outputs $w_1$.
			
			$P_2$: checks whether $\bar{x}_1\in {\mathbb{F}_q}^k$ and $r_1\in \mathbb{F}_q$:
			\begin{itemize}
				\item if $\bar{x}_1\notin {\mathbb{F}_q}^k$, or $r_1\notin \mathbb{F}_q$ aborts and outputs $w_2=\lambda$,
				\item otherwise, outputs $w_2= \langle \bar{x}_1 \cdot \bar{y}_0 \rangle +r_1-s_0$.
			\end{itemize}
			
		\end{minipage}
		}
	\caption{The DIP protocol of \cite{david2016unconditionally}
						\label{fig:DDG+ SSP protocol}}
\end{figure}

\section{Security Analysis}\label{sec:security analysis}

A careful inspection of protocol $\pi_{DDG+}$ reveals that it does not simulate the ideal $\mathcal{F}_{DIP}$ completely. It is easy to check this for $k=1$, where the input vectors reduce to scalar values. In this case, at the end of a run of protocol $\pi_{DDG+}$, $P_2$ learns $P_1$’s input. In \cite{dowsley2010two}, authors ignore this case and reason that DIP functionality inherently is not private for $ k=1 $. As intuitively mentioned earlier in Remark \ref{Rem:SSPandCSP}, it is the case for conventional inner product and not for DIP. In the following proposition, we formally state an almost trivial security property of $\mathcal{F}_{DIP}$ which flatly contradicts the reasons given in \cite{dowsley2010two}.

\begin{proposition}\label{prop:SSP security for k=1}
	If $\pi$ is a protocol that securely computes $\mathcal{F}_{DIP}$ for input vectors of length $k$, then for every $1\leq k'\leq k$, there is a protocol $\pi_{k'}$ constructed with black-box use of $\pi$ which securely computes $\mathcal{F}_{DIP}$ for input vectors of length $k'$.
\end{proposition}

\begin{proof}
	$P_1$ and $P_2$ perform protocol $\pi_{k'}$ on inputs $\bar{x}=(x_1,x_2,\ldots,x_{k'})$ and $\bar{y}=(y_1,y_2,\ldots,y_{k'})$ as follows:
	\begin{enumerate}
		\item They append a zero vector of length $k-k'$ to their inputs and form extended inputs as $\bar{x}_e=(x_1,\ldots,x_{k'},0,\ldots,0)$ and $\bar{y}_e=(y_1,\ldots,y_{k'},0,\ldots,0)$.
		\item Then, they run $\pi$ on the extended inputs $\bar{x}_e$, $\bar{y}_e$, and receive $w_1$ , $w_2$.
		\item They output $w_1$ , $w_2$.
	\end{enumerate}
	It is easy to check that correctness and security of protocol $\pi_{k'}$ reduces to correctness and security of protocol $\pi$.
\end{proof}

Now, it can be concluded that protocol $\pi_{DDG+}$ does not realize $\mathcal{F}_{DIP}$ for $k=1$ and therefore it is not secure. Protocol $\pi_{DDG+}$ is insecure for $k>1$ as well. Particularly, $P_2$ can always deduce an equation on $P_1$'s input vector, $\bar{x}$, after receiving $\bar{x}_1=\bar{x}+\bar{x}_0$. That is,

\begin{gather*}
	\bar{x}_1=\bar{x}+\bar{x}_0\\
	\overset{\cdot \bar{y}_0}{\Longrightarrow} \langle \bar{x}_1 \cdot \bar{y}_0 \rangle=\langle \bar{x} \cdot \bar{y}_0 \rangle+\langle \bar{x}_0 \cdot \bar{y}_0 \rangle\\
	\Longrightarrow \langle \bar{x} \cdot \bar{y}_0 \rangle = \langle \bar{x}_1 \cdot \bar{y}_0 \rangle - s_0
\end{gather*}

Note that $P_2$ is not able to deduce such an equation in the ideal run of DIP functionality. Therefore, Protocol $\pi_{DDG+}$ leaks some information about $P_1$'s input beyond what is available in the ideal world. Technically, the ideal world simulator fails to construct the view of a real world adversary that controls $P_2$. It is important to note that the simulation failure happens considering both semi-honest and malicious adversaries. Hence Protocol $\pi_{DDG+}$ is insecure in both adversarial models.

The information leakage explained above, may seem harmless for large vectors because the adversary learns virtually nothing about $P_1$'s input. But, regarding composition, this little information leakage will be a serious problem. We demonstrate it in an example scenario.

\subsection{Example scenario: multiplication of a vector by a matrix}
Suppose that $P_1$ has a $1\times k$ vector $\bar{x}\in {\mathbb{F}_q}^k$ and $P_2$ has a $k \times k$ matrix $\textbf{Y}\in {\mathbb{F}_q}^{k \times k}$ and they intend to compute $\bar{w}=\bar{x}\times \textbf{Y}$ in a shared manner. That is, they want to receive random shares $\bar{w}_1$ and $\bar{w}_2$, correspondingly, such that $\bar{w}=\bar{w}_1+\bar{w}_2$. This computation will be trivial provided that $P_1$ and $P_2$ can run an arbitrary number of instances of a universally composable protocol for computing $\mathcal{F}_{DIP}$.

It is easy to check that protocol $\pi_{DDG+}$ is completely insecure to be the underlying DIP protocol for this scenario. Assume that $P_1$ and $P_2$ run protocol $\pi_{DDG+}$ to compute each element of $\bar{w}$. We denote the $i$th element of $\bar{w}$ by $w_i$ and the $i$th column of $\textbf{Y}$ by $Y_i$. To compute $w_i$, $P_1$ inputs $\bar{x}$ and $P_2$ inputs $Y_i$ in protocol $\pi_{DDG+}$. The output of each party is a share of $w_i$. In addition to its output share, $P_2$ derives an equation on $\bar{x}$ from each run of protocol $\pi_{DDG+}$ as discussed above. After $k$ runs, the parties receive their desired outputs. But, in this state $P_2$ has $k$ equations on $\bar{x}$ that will be enough to extract $\bar{x}$ if they are linearly independent. Specifically, if we denote by $\bar{x}_0^i$, $\bar{y}_0^i$ and $s_0^i$ the randomness received by the corresponding parties in the preprocessing phase of the $i$th run of protocol $\pi_{DDG+}$, then the equation that $P_2$ can deduce from the $i$th run will be of the form
\begin{equation*}
	\langle \bar{x} \cdot \bar{y}_0^i \rangle = \langle \bar{x}_1^i \cdot \bar{y}_0^i \rangle + s_0^i,
\end{equation*}.

\noindent where $\bar{x}_1^i=\bar{x}+\bar{x}_0^i$ is the message that $P_2$ receives from $P_1$ in the $i$th run. Let $\textbf{Y}_0$ be the matrix with $\bar{y}_0^i$ as its $i$th column and $\bar{q}_0$ be the vector with $\langle \bar{x}_1^i \cdot \bar{y}_0^i \rangle + s_0^i$ as its $i$th element. Now, we can write the set of these equations in the form of $\bar{x} \times \textbf{Y}_0=\bar{q}_0$ which is a system of linear equations. Therefore, vector $\bar{x}$ will be found uniquely if $\textbf{Y}_0$ is a non-singular matrix.

It is easy to extend above argument to matrix multiplication scenario and therefore all linear algebra protocols of \cite{david2016unconditionally}.

\section{Conclusion}\label{sec:conclusion}
Inspection of the proof given in \cite{david2016unconditionally} and \cite{dowsley2010two} for security of protocol $\pi_{DDG+}$ reveals that, in the simulation process, the case of corrupted $P_2$ is not discussed in detail due to the apparent similarity to the case of corrupted $P_1$. Considering the presented security flaw in this protocol, we can conclude a general recommendation to avoid any frugality in the process of security proofs, especially when tasks of parties are not exactly identical in the protocol.

In \cite{david2016unconditionally} and \cite{dowsley2010two}, various higher level privacy preserving linear algebra protocols are proposed that the DIP protocol is their fundamental building block. Fortunately, these higher level protocols use the DIP protocol as a black-box and, thus, they can be implemented using any universally composable secure DIP protocol.

\bibliography{Refs}

\end{document}